\documentclass[11pt]{article}

\usepackage[margin=1in]{geometry}
\usepackage{microtype}
\usepackage{xcolor}

\usepackage{enumitem}
\usepackage{graphicx}
\usepackage{fancybox}
\usepackage{comment}
\usepackage{xcolor}
\usepackage{nameref}
\usepackage{bbm}

\definecolor{ForestGreen}{rgb}{0.1333,0.5451,0.1333}
\definecolor{DarkRed}{rgb}{0.8,0,0}
\definecolor{Red}{rgb}{1,0,0}
\usepackage[linktocpage=true,
colorlinks,
linkcolor=ForestGreen,citecolor=ForestGreen,
bookmarks,bookmarksopen,bookmarksnumbered]
{hyperref}

\usepackage{amsmath}
\usepackage{amssymb}
\usepackage{amsthm}
\usepackage{mathtools}

\usepackage{subfig}

\usepackage{thm-restate}
\usepackage[ruled, noend, linesnumbered]{algorithm2e}

\usepackage{float}
\usepackage{cleveref}

\def\setof#1{\left\{#1  \right\}}

\def\ceil#1{\left\lceil #1 \right\rceil}

\def\dim#1{\mathrm{dim} (#1)}

\def\norm#1{\left\| #1 \right\|}

\def\smallnorm#1{\| #1 \|}

\def\normop#1{\left\| #1 \right\|_{\mathrm{op}}}

\def\OL{\operatorname{OL}}

\def\calA{\mathcal{A}}

\def\calS{\mathcal{S}}

\newcommand\vecone{\boldsymbol{1}}

\newcommand\matzero{\boldsymbol{0}}

\newcommand\poly{\textnormal{poly}}

\newcommand\R{\mathbb{R}}

\newcommand\tr{\mathrm{Tr}}

\DeclareMathOperator*{\diag}{diag}

\newcommand{\polylog}{\text{ polylog }}

\newcommand{\dist}{\textnormal{dist}}

\newcommand{\eps}{\varepsilon}

\newcommand{\pe}{\preceq}

\newcommand{\supp}{\mathrm{supp}}

\newcommand{\vone}{\mathbf{1}}

\renewcommand{\tilde}{\widetilde}

\newcommand{\cF}{\mathcal{F}}

\newcommand{\Otil}{\tilde{O}}

\renewcommand{\epsilon}{\ensuremath\varepsilon}

\renewcommand{\phi}{\ensuremath{\varphi}}

\renewcommand\AA{\boldsymbol{\mathrm{{A}}}}
\newcommand\BB{\boldsymbol{\mathrm{{B}}}}
\newcommand\CC{\boldsymbol{\mathrm{{C}}}}

\newcommand\EE{\boldsymbol{\mathrm{{E}}}}

\newcommand\HH{\boldsymbol{\mathrm{{H}}}}
\newcommand\II{\boldsymbol{\mathrm{{I}}}}

\newcommand\KK{\boldsymbol{\mathrm{{K}}}}
\newcommand\MM{\boldsymbol{\mathrm{{M}}}}

\newcommand\LL{\boldsymbol{\mathrm{{L}}}}

\newcommand\WW{\boldsymbol{\mathrm{{W}}}}
\newcommand\VV{\boldsymbol{\mathrm{{V}}}}
\newcommand\XX{\boldsymbol{\mathrm{{X}}}}
\newcommand\YY{\boldsymbol{\mathrm{{Y}}}}
\newcommand\ZZ{\boldsymbol{\mathrm{{Z}}}}

\renewcommand\aa{\boldsymbol{\mathrm{a}}}
\newcommand\bb{\boldsymbol{\mathrm{b}}}

\newcommand\vv{\boldsymbol{\mathrm{v}}}
\newcommand\ww{\boldsymbol{\mathrm{w}}}
\newcommand\xx{\boldsymbol{\mathrm{x}}}

\SetKw{Break}{break}
\SetKwComment{Comment}{$\triangleright$\ }{}
\SetCommentSty{mycommfont}

\newtheorem{theorem}{Theorem}[section]

\newtheorem{lemma}[theorem]{Lemma}
\newtheorem{observation}[theorem]{Observation}

\newtheorem{definition}[theorem]{Definition}
\newtheorem{remark}[theorem]{Remark}

\newtheorem*{theorem*}{Theorem}

\newtheorem*{corollary*}{Corollary}
\newtheorem*{conjecture*}{Conjecture}
\newtheorem*{lemma*}{Lemma}
\newtheorem*{thm*}{Theorem}
\newtheorem*{prop*}{Proposition}
\newtheorem*{obs*}{Observation}
\newtheorem*{definition*}{Definition}

\newtheorem*{remark*}{Remark}
\newtheorem*{rec*}{Recommendation}

\usepackage[backend=biber, isbn=false, style=alphabetic, backref=true, doi=true, url=false, eprint=false, maxcitenames=10, mincitenames=3, maxbibnames=10, minbibnames=6, minalphanames=3, maxalphanames=5, defernumbers=true]{biblatex}

\usepackage{makecell}
\usepackage{booktabs}
\usepackage{array}
\usepackage{tablefootnote}

\newcommand{\E}{\mathbb E}
\newcommand{\Prb}{\mathbb P}

\title{An Online Sparsification Algorithm from the Book}
\author{
  Gramoz Goranci \\
  \small University of Vienna \\
  \small \texttt{gramoz.goranci@univie.ac.at}
  \and
  Rasmus Kyng\thanks{The research leading to these results has received funding from the starting grant “A New Paradigm for Flow and Cut Algorithms” (no. TMSGI2 218022) and grant no. 200021 204787 of the Swiss National Science Foundation.} \\
  \small ETH Zurich \\
  \small \texttt{kyng@inf.ethz.ch}
  \and
  Maximilian Probst Gutenberg\footnotemark[1] \\
  \small ETH Zurich \\
  \small \texttt{maximilian.probst@inf.ethz.ch}
  \and
  Yibin Zhao \\
  \small University of Toronto \\
  \small \texttt{ybzhao@cs.toronto.edu}
  \and
  Gernot Z\"ocklein\footnotemark[1] \\
  \small ETH Zurich \\
  \small \texttt{gernot.zoecklein@inf.ethz.ch}
}

\date{\today}

\begin{document}
\maketitle

\begin{abstract}
In their seminal paper \cite{cohen2016online}, Cohen, Musco, and Pachocki proposed a natural and simple online spectral sparsification algorithm: rows $\aa_1, \aa_2, \ldots \in \mathbb{R}^d$ of a matrix $\AA$ arrive one-by-one, and when row $\aa_i$ arrives, it is appended to sparsifier $\tilde{\AA}$ (after appropriately reweighting it) with probability proportional to its \emph{current} leverage score
\[
\tau^{\mathrm{OL}}(\aa_i)=\aa_i^\top(\AA_i^\top\AA_i)^\dagger \aa_i,  \text{ where }\AA_i = [\aa_1, \aa_2, \ldots, \aa_i]^\top
\]
or otherwise discarded forever. For oblivious streams, they showed that this maintains a $(1\pm\epsilon)$-spectral approximation $\tilde{\AA}$ of every $\AA$ with $O(d\epsilon^{-2}\log^2 d)$ many rows.

A natural question is whether the same algorithm works for adaptive streams, where each row may depend on the algorithm’s previous random choices. The original proof does not extend directly: it analyzes the process in isotropic position with respect to the final matrix $\AA$, which is not fixed in advance under adaptivity. As an extension of this proof framework remained elusive, various algorithmic variants have since been suggested.

In this paper, we show that the original online leverage-score sampling algorithm is indeed robust to adaptive adversaries. Our main technical contribution is a Freedman-type matrix martingale inequality with an evolving isotropic map, allowing the isotropic map used in the concentration argument to change with the stream.

As a consequence, this gives the first online sparsification algorithm for adaptive streams that yields a sparsifier of near-optimal size $O(d \eps^{-2}\log^2 d)$ whose working memory is proportional to the size of the sparsifier. For the special case of spectral graph sparsification, we provide an implementation that additionally runs in time near-linear in the stream size.
\end{abstract}

\pagebreak

\paragraph{Acknowledgements.}
The initial proof of the matrix concentration  \Cref{lem:fixed-freedman} was generated in a conversation with ChatGPT $5.5$ Pro. The authors have verified
and rewritten the proof to improve readability and presentation.

\section{Introduction}

A fundamental primitive in numerical linear algebra is to compute a small spectral approximation to a tall matrix $\AA \in \R^{n \times d}$. We say that $\Tilde{\AA} \in \R^{k \times d}$ is a $(1\pm \eps)$-spectral approximation to $\AA$ if
$$
(1-\eps)\AA^\top \AA \preceq \Tilde{\AA}^\top \Tilde{\AA} \preceq (1+\eps)\AA^\top \AA .
$$
Equivalently, $\smallnorm{\Tilde{\AA}\xx}_2^2$ approximates $\norm{\AA\xx}_2^2$ for every vector $\xx \in \R^d$. If $k \ll d$, then $\Tilde{\AA}$ can usually be used as a sparse proxy of $\AA$ in downstream computations to improve runtimes and memory usage. Very often, it is additionally desirable that the matrix $\Tilde{\AA}$ is built by re-weighting existing rows of $\AA$. This not only helps with interpretability, but it also respects the structure of the data. In this paper, we are interested in the case where $\Tilde{\AA}$ has to be a re-weighting of $\AA$.

The standard way to find such a spectral approximation $\Tilde{\AA}$ of $\AA$ is by \emph{leverage score sampling}, as introduced in \cite{SpielmanS08}. For a row $\aa_i$, define its leverage score as 
$$\tau(\aa_i) := \aa_i^\top (\AA^\top \AA)^\dag \aa_i.$$
Then, in \cite{SpielmanS08}, it was shown that for some $\rho = O(\eps^{-2} \log d)$, adding every row with probability $p_i = \min\setof{1, \rho \cdot \tau(\aa_i)}$ to $\tilde{\AA}$ scaled by $1/p_i$ yields a $(1 \pm \eps)$-approximate spectral sparsifier $\tilde{\AA}$ of $\AA$ with high probability. Furthermore, the expected number of sampled rows is $O(d \eps^{-2}\log d)$, which is optimal up to the $\log d$ factor.

\paragraph{Online Sparsification.} In this article, we consider the online model where rows $\aa_1, \aa_2, \ldots, \aa_n$ arrive one-by-one. Upon arrival of row $\aa_i$, we have to irrevocably decide whether to keep or discard the row. A natural and simple algorithm for this setting is to define $\AA_{i} = [\aa_1 \; \aa_2 \; \ldots \; \aa_i]^{\top}$, and upon receiving row $\aa_i$ to compute the \emph{current} leverage score
\[
    \tau^{\operatorname{OL}}(\aa_i) := \aa_i^\top (\AA_i^\top \AA_i)^\dag \aa_i,
\]
and to maintain a sparsifier $\tilde{\AA}$ of $\AA$ by successively adding each row $\aa_i$ upon arrival to the sparsifier with probability  $p^{\operatorname{OL}}_i = \min\setof{1, \rho \cdot \tau^{\operatorname{OL}}(\aa_i)}$ scaled by $1/\sqrt{p^{\operatorname{OL}}_i}$. 

Assuming that the insertion sequence of rows $\aa_i$ is prefixed, it is not hard to extend the by-now standard proof by Spielman-Srivastava \cite{SpielmanS08} to show that at any time, $\tilde{\AA}$ is a $(1\pm\epsilon)$-spectral sparsifier of $\AA$, with high probability. As shown in \cite{cohen2016online}, using leverage scores $\tau^{\operatorname{OL}}(\aa_i)$ in lieu of $\tau(\aa_i)$ only increases the number of samples in the final sparsifier to $O(d \eps^{-2}\log^2 d)$ rows, which is optimal in the online setting up to a $\log d$ factor\footnote{For simplicity, we assume in the abstract and introduction that $A$ contains only polynomially bounded integer entries, and that the stream length $n$ is polynomially bounded in the dimension $d$.}.

\paragraph{Robust Online Sparsification.} While prefixed insertion sequences seem natural in many settings, such as reading blocks from memory, there are also many settings where the insertion sequence depends on intermediate outputs of the sparsifier, i.e. where new rows are generated on-the-fly. A streaming service might maintain a sparsifier $\tilde{\AA}$ of a user's data $\AA$ and based on this information make recommendations that then lead to new user data; a navigation system might suggest an uncongested route to many users based on traffic information which might then congest it; etc. 

In this more challenging setting, one requires the algorithm to maintain $\tilde{\AA}$ robustly, i.e. to remain a spectral approximation of $\AA$ even under adversarial updates. Naturally, one might ask whether the above online sparsification algorithm extends to the robust setting. 

Unfortunately, the proof by Spielman-Srivastava \cite{SpielmanS08} can no longer apply here. The key issue is that Spielman-Srivastava argues via Matrix Concentration of the normalized matrix $(\AA^{\top} \AA)^{\dag / 2}\tilde{\AA}^{\top}\tilde{\AA}(\AA^{\top} \AA)^{\dag / 2}$; however, the isotropic map $(\AA^{\top} \AA)^{\dag / 2}$ depends on the final matrix $\AA$, which is itself random under adaptivity. This cannot be handled by standard black box matrix concentration inequalities.

As a result, various new algorithms were proposed recently to maintain a spectral sparsifier $\tilde{\AA}$ in the online, streaming and dynamic graph settings \cite{rowadaptive,kennethmordoch2025adversarialrobustnessonlineimportance,shunhua_leastsquares, BernsteinvdBGNSSS22}.

\paragraph{Our Contribution.} In this article, we give a natural generalization of the underlying Matrix Concentration bounds used by Spielman-Srivastava \cite{SpielmanS08}. As an immediate result, we obtain that indeed, the natural Online Sampling algorithm succeeds even when new rows are adversarially generated on-the-fly.

\begin{theorem}\label{thm:mainThm}
Given $\eps > 0$ and a matrix $\AA \in \mathbb{R}^{n \times d}$ whose rows arrive as an adaptive stream $\aa_1, \aa_2, \ldots, \aa_n$. Consider the algorithm where, upon arrival, row $\aa_i$ is added (after reweighting) to the sparsifier $\tilde{\AA}$ with probability $p_i^{OL}$. Then, w.h.p., for any time $i$, for $\AA_i$ being the prefix of the first $i$ rows of $\AA$ and $\tilde{\AA}_i$ the sparsifier at the time, we have 
    \[ (1- \eps) \AA_i^\top \AA_i \preceq \Tilde{\AA}_i^\top \Tilde{\AA}_i \preceq (1+\eps) \AA_i^\top \AA_i. \]
The sparsifier $\tilde{\AA}$ has at most $O(d \epsilon^{-2} \log^2 d)$ rows at any time. Moreover, the working memory required is proportional to the size of the sparsifier.
\end{theorem}

Our algorithm matches the adversarially robust derandomized online BSS algorithm from \cite{braverman2020near} in the size of the sparsifier it produces. Additionally, for matrices $A$ with $k$-sparse rows, the working memory\footnote{We consider the real RAM memory model here.} of our algorithm is $O(k d \epsilon^{-2} \log^2 d)$ in comparison to their $O(d^2 + k d \epsilon^{-2} \log^2 d)$. This is particularly significant for the special case of $\AA$ being the incidence matrix of a graph, which is 2-sparse. Here we show that due to the simplicity of our algorithm, it can be implemented to run in near-linear time with near-optimal working memory --- essentially achieving the best of all worlds.

\begin{theorem}\label{thm:mainThmRuntime}
Given $\eps > 0$ and a matrix $\AA \in \mathbb{R}^{n \times d}$ where $\AA^\top \AA$ is a Graph Laplacian and whose rows arrive as an adaptive stream $\aa_1, \aa_2, \ldots, \aa_n$. Then the above algorithm in \Cref{thm:mainThm} can be implemented to maintain a $(1\pm\eps)$-spectral sparsifier with at most $\tilde{O}(d \eps^{-2})$ rows, while using at most at most $\tilde{O}(n)$ total processing time and at most $\Otil(d \eps^{-2})$ working memory.\footnote{In this paper, we use the notation $f(d) = \Otil(g(d))$ if $f(d) = O(g(d)\cdot\polylog d)$.}
\end{theorem}

In \Cref{fig:online-spectral-approximation} we compare the performance of our algorithm to other online and streaming algorithms. In \Cref{fig:online-spectral-sparsification} we compare our implementation of online graph sparsification to other available results.

\begin{table}[H]
\centering
\renewcommand{\arraystretch}{1.12}
\setlength{\tabcolsep}{8pt}
\begin{tabular}{r|c|c|c|c}
Reference
& Online
& Number of rows
& Robust
& Memory ($k$-sparse rows) \\
\hline
\cite{kennethmordoch2025adversarialrobustnessonlineimportance}
& No
& $\Otil(d \epsilon^{-2})$
& Yes
& $\Otil(kd \epsilon^{-2})$
\\
\cite{KPPS}
& No
& $O(d \epsilon^{-2} \log d)$
& No 
& $O(kd \epsilon^{-2} \log d)$
\\
\small{Online BSS} \cite{cohen2016online}
& Yes
& $O(d \epsilon^{-2} \log d)$
& No
& $O(k d \epsilon^{-2} \log^2 d + d^2)$ \\
 \cite{braverman2020near}
& Yes
& $O(d \epsilon^{-2} \log^2 d)$
& Yes
& $O(k d \epsilon^{-2} \log^2 d + d^2)$ \\
\Cref{thm:main2}
& Yes
& $O(d \epsilon^{-2} \log^2 d)$
& Yes
& $O(k d \epsilon^{-2} \log^2 d)$ \\
\end{tabular}
\caption{Comparison of online and streaming algorithms for spectral approximation.}
\label{fig:online-spectral-approximation}
\end{table}

\begin{table}[H]
\centering
\renewcommand{\arraystretch}{1.12}
\setlength{\tabcolsep}{8pt}
\begin{tabular}{r|c|c|c|c|c}
Reference
& Online
& Number of edges
& Robust
& Run-time
& Memory \\
\hline
\cite{cohenaddad2025nearlyspaceoptimalgraphhypergraph}
& No
& $O(n \epsilon^{-2} \poly(\log \log n))$
& No
& $\poly(n)$
& $O(n \epsilon^{-2} \poly(\log \log n))$
\\
\cite{KPPS}
& No
& $O(n \epsilon^{-2} \log n)$
& No
& $\Otil(m)$
& $\Otil(n \epsilon^{-2})$
\\
\cite{kennethmordoch2025adversarialrobustnessonlineimportance}
& No 
& $\Otil(n \epsilon^{-2})$
& Yes
& $\poly(n)$ & 
$\Otil(n \epsilon^{-2})$
\\
\small{Online BSS} \cite{cohen2016online}
& Yes
& $O(n \epsilon^{-2} \log n)$
& No
& $\Omega(n^2)$
& $\Omega(n^2)$
\\
\cite{braverman2020near}
& Yes
& $O(n \epsilon^{-2} \log^2 n )$
& Yes
& $\Omega(n^{2})$ 
& $\Omega(n^2)$
\\
\cite{AbrahamDKKP16}
& Yes
& $\Otil(n \epsilon^{-2})$
& No
& $\Otil(m)$
& $\Otil(n \epsilon^{-2})$ \\
\Cref{thm:graph-sparsifier}
& Yes
& $\Otil(n \epsilon^{-2})$
& Yes
& $\Otil(m)$
& $\Otil(n \epsilon^{-2})$ \\
\end{tabular}
\caption{Comparison of online and streaming algorithms for spectral graph sparsification.}
\label{fig:online-spectral-sparsification}
\end{table}


\paragraph{An Independent Work.} An independent proof of the robustness of leverage scores sampling was recently posted publicly on GitHub \cite{peng2026leverage}.

\section{Preliminaries}
\paragraph{Vectors and Matrices.}
For matrix $\MM \in \R^{k \times d}$, we denote by $\MM^\dag$ its Moore-Penrose pseudo-inverse.
We often denote $\AA \in \R^{n \times d}$ as an input matrix and $\aa_i \in \R^d$ by its rows as column vectors for $i \in [n]$.
For $i \in [n]$, we use $\AA_i$ to denote the $i \times d$ matrix obtained by the first $i$ rows of $\AA$.
We denote by $\supp(\cdot)$ the support of vector.

We use $\pe$ to denote the Loewner partial order of symmetric matrices.
For a symmetric matrix $\MM \in \R^{d \times d}$, we let $\lambda_i(\MM)$ denote the $i$th smallest eigenvalue of $\MM$, so $\lambda_1(\MM) \le \lambda_2(\MM) \le \ldots \le \lambda_d(\MM)$.
We also use $\lambda_{\max}(\MM)$ and $\lambda_{\min}(\MM)$ to denote the largest and smallest eigenvalues respectively.
We use $\normop{\cdot}$ to denote the operator norm.

\paragraph{Graphs.}
We denote an undirected weighted graph by $G=(V,E,\ww)$ with vertices $V$, edges $E$, and edge weights $\ww \in \R_{\ge 0}^E$.
We say $H$ is a subgraph of $G$ if the edges and vertices of $H$ are subsets of the edges and vertices of $G$.

For an undirected (possibly weighted) graph $G$, we say a subgraph $H$ is an $\alpha$-spanner of $G$ if for every pair of vertices $u,v \in V$, it holds that $d_H(u,v) \le \alpha \cdot d_G(u,v)$ in the shortest path metric $d$.

\paragraph{Graph matrices.}
For an undirected weighted graph $G=(V,E,\ww)$, we let $\BB \in \{-1,0,1\}^{E \times V}$ be its edge-vertex incident matrix with an arbitrarily fixed orientation of the edges.
For an edge $e=(u,v)$ with orientation from $u$ to $v$, we define $\bb_e = \bb_{uv} = \vone_u - \vone_v$.
When clear from context that $\ww$ are edge weights, we let $\WW = \diag(\ww)$ be the diagonal matrix in $\R_{\ge 0}^{E \times E}$.
Note that $\bb_e$ corresponds to the row of $\BB$ labelled by $e \in E$.
The Laplacian matrix of $G$ is defined by $\LL_G = \BB^\top \WW \BB$.

\paragraph{Concentration inequalities.}
We state the following (slightly modified version) of a scalar Freedman inequality. A proof can straightforwardly be obtained by slightly modifying the standard proof of Freedman's inequality. For a proof of a generalization of this result, see Theorem $1$ in \cite{chernoff_time_unif}.
\begin{lemma}[\cite{freedman}]\label{lma:scalar_freedman}
Let $(\mathcal F_i)_{i=1}^n$ be a filtration, and let
$X_1,\dots,X_n$ be $\{0,1\}$-valued random variables such that
$X_i$ is $\mathcal F_i$-measurable. Define
\[
    p_i=\mathbb E[X_i\mid \mathcal F_{i-1}],
    \qquad
    S_k=\sum_{i=1}^k X_i,
    \qquad
    \mu_k=\sum_{i=1}^k p_i .
\]
Then, for every $x>0$,
\[
     \Pr\left[
        \exists k\le n:
        S_k-\mu_k \ge \mu_k+x
    \right]
    \le e^{-x}.
\]
In particular,
\[
    \Pr\left[
        S_n \ge 2\mu_n+x
    \right]
    \le e^{-x}.
\]
\end{lemma}

\paragraph{Misc.} We denote by $i \land T = \min \setof{i, T}$. We assume $n \geq d$, and $n \geq 2$.

\section{Adaptive Online Row Sampling}
\subsection{Main Theorem \& Algorithm}
In this section, we provide the pseudo-code of the online row sampling algorithm and state our main theorems. In \Cref{algo:main_algo}, we provide the general algorithmic framework that yields spectral approximations against an adaptive adversary. We then proceed to show how to get very sparse approximations using memory proportional to the size of the approximation. Later, in \Cref{sec:graphs}, we additionally provide a fast implementation for the special case of spectral graph sparsification. 

Additionally, there are two cases we will have to distinguish, depending on the assumptions on the data stream. The first case is the general one, where we have no restrictions on the incoming rows $\aa_1, \aa_2, \ldots, \aa_n \in \R^d$. In this case, the final number of sampled rows scales proportional to $\log \kappa^{\operatorname{OL}}$, a type of online condition number defined as
$$\kappa^{\operatorname{OL}} := \norm{\AA_n}_{\mathrm{op}} \cdot (\max_{i=1}^n \smallnorm{\AA_i^{\dag}}_{\mathrm{op}}).$$
This condition number can be very large, and it can be shown that a dependence of $\log \kappa^{\operatorname{OL}}$ in the number of samples is necessary for online sparsification \cite{cohen2016online}. Fortunately, in \cite{woodruff2022high}, it was shown that if one restricts the rows to have values that are polynomially bounded integers, then one can replace the $\log \kappa^{\operatorname{OL}}$ factor by a simple $\log n$ factor. 

Before we state our theorems, we first state the following bound on the online leverage scores.
\begin{theorem}[See \cite{cohen2016online} and \cite{woodruff2022high}]\label{thm:online_lev_sum}
    Let $\AA \in \mathbb{R}^{n \times d}$. Then
    \[ 
    \sum_{i=1}^n \tau^{\OL}(\aa_i) = O(d \log \kappa^{\operatorname{OL}}).
    \]
    Furthermore, if $\AA$ is assumed to have polynomially bounded integer entries, the above sum is instead bounded by $O(d \log n)$.
\end{theorem}

We can now state the main theorem, which shows that over-sampling according to online leverage scores yields a spectral approximation with high probability. Notice that we do not actually sample according to $\tau^{\operatorname{OL}}(\aa_i)$ but rather to an estimate of it that uses the current random sample $\Tilde{\AA}_{i-1}$. We provide pseudo-code of the algorithm in Algorithm \ref{algo:main_algo}. We will prove the theorem below later in \Cref{sec:thm_proof}.

\begin{theorem}\label{thm:main}
    Given $\eps > 0$ and a matrix $\AA \in \mathbb{R}^{n \times d}$, whose rows arrive in a stream $\aa_1, \aa_2, \ldots, \aa_n$. Then, the matrices $\Tilde{\AA}_i$ produced by \Cref{algo:main_algo} satisfy that for all $i$,   
    \[ (1- \eps) \AA_i^\top \AA_i \preceq \Tilde{\AA}_i^\top \Tilde{\AA}_i \preceq (1+\eps) \AA_i^\top \AA_i \]
    with probability at least $1 - 2n^{-1}$. 
\end{theorem}

\begin{algorithm}[H]
\caption{$\textsc{OnlineSample}(\AA, \eps)$}
\label{algo:main_algo}
\textbf{Input}: Adversarial stream $\aa_1, \aa_2, \ldots, \aa_n \in \R^d$ of rows. \\
\textbf{Output}: $(1 \pm \eps)$-spectral approximation $\Tilde{\AA}_i$ of $\AA_i$ \\ 
$\rho \gets 8 \eps^{-2} \log n$;  // If $n$ is unknown, any over-estimate suffices \\
$\Tilde{\AA}_0 \gets 0 \in \mathbb{R}^{0 \times d}$; \\
\For{$i = 1, 2, \ldots, n$}{
    Let $p_i \geq \min \setof{1, \rho \cdot (1+\eps) \cdot \aa_i^\top (\Tilde{\AA}_{i-1}^\top \Tilde{\AA}_{i-1} + \aa_i \aa_i^{\top})^\dag \aa_i}$ \label{algo:main_algo:line6}\;
    $
    \Tilde{\AA}_i \gets \begin{cases}
        \begin{bmatrix}
            \Tilde{\AA}_{i-1} \\
            \aa_i / \sqrt{p_i} 
        \end{bmatrix} & \textnormal{with probability } p_i \\
        \Tilde{\AA}_{i-1} & \textnormal{else}
    \end{cases}
    $
}
\end{algorithm} 
\begin{remark}
    We note that in Line $6$ of \Cref{algo:main_algo:line6} one can also instead directly choose sampling probabilities $p_i \geq \rho \tau_i^{\operatorname{OL}}(\aa_i)$ if available. 
\end{remark}

Using the above two theorems, we can straightforwardly show that by setting 
$$p_i = \min \setof{1, \rho \cdot (1+\eps) \cdot \aa_i^\top (\Tilde{\AA}_{i-1}^\top \Tilde{\AA}_{i-1} + \aa_i \aa_i^{\top})^\dag \aa_i},$$ the resulting matrices $\Tilde{\AA}_i$ will also be very sparse with high probability.

\begin{theorem}\label{thm:main2}
    Given $\eps > 0$ and a matrix $\AA \in \mathbb{R}^{n \times d}$, whose rows arrive in a stream $\aa_1, \aa_2, \ldots, \aa_n$, there is an adversarially robust online algorithm such that for all $i$, 
    \[ (1- \eps) \AA_i^\top \AA_i \preceq \Tilde{\AA}_i^\top \Tilde{\AA}_i \preceq (1+\eps) \AA_i^\top \AA_i. \]
    The algorithm samples $O(d \eps^{-2} \log n \log \kappa^{\operatorname{OL}})$ many rows. It is correct with probability at least $1-3n^{-1}$, and requires memory proportional to the size of the approximation.

    Moreover, if we assume that $A$ only has polynomially bounded integer entries, one may replace the $\log \kappa^{\operatorname{OL}}$ factors above by $\log n$. 
\end{theorem}
\begin{proof}
    We show that runing \Cref{algo:main_algo} with $p_i = \min \setof{1, \rho \cdot (1+\eps) \cdot \aa_i^\top (\Tilde{\AA}_{i-1}^\top \Tilde{\AA}_{i-1})^\dag \aa_i}$ yields the desired guarantees.
    
    So let $(\mathcal{F}_i)_{i=1}^n$ be the filtration where $\mathcal{F}_i$ includes all random bits used by the algorithm after the $i$th row has been processed. Note that we may assume that the adversary is deterministic. Consequently, once $\cF_{i-1}$ is fixed, $\aa_i$ is deterministic, and it is $\mathcal{F}_{i-1}$-measurable.

    Let $\tilde{\tau}_i := \aa_i^\top (\Tilde{\AA}_{i-1}^\top \Tilde{\AA}_{i-1} + \aa_i \aa_i^{\top})^\dag \aa_i$. Let $\xi_i \sim \operatorname{Ber}(p_i)$, where $p_i = \min \setof{1, \rho \cdot (1+\epsilon) \tilde{\tau}_i}$. Note that $\xi_i$ is $\mathcal{F}_i$-measurable, and that $\tilde{\tau}_i, p_i$ are $\mathcal{F}_{i-1}$-measurable. Let $N_i = \sum_{j=1}^i \xi_j$ denote the number of sampled rows after processing the first $i$ rows. Consider the event 

    $$\mathcal{E} := \setof{N_n \leq 2 \log n + 2 \sum_{i=1}^n p_i} \cap \setof{\forall i: (1- \eps) \AA_i^\top \AA_i \preceq \Tilde{\AA}_i^\top \Tilde{\AA}_i \preceq (1+\epsilon) \AA_i^\top \AA_i}.$$
    
    Then proving the correctness of the algorithm is equivalent to proving that 

    $$
    \Pr[\mathcal{E}] \geq 1 - 3n^{-1}.
    $$
    
    To see why, note that as long as $(1- \eps)\AA_i^\top \AA_i \preceq \Tilde{\AA}_i^\top \Tilde{\AA}_i$, we have that $\tilde{\tau}_i \leq 1/(1- \eps) \tau^{\OL}(\aa_i)$, i.e., that $p_i \leq \frac{\rho (1+\epsilon)}{1-\epsilon} \tau^{\OL}(\aa_i)$. Combining this with \Cref{thm:online_lev_sum} shows that
    \[
    N_n \vecone_{\mathcal{E}} \leq 2 \log n + 2 \sum_{i=1}^n p_i \vecone_{\mathcal{E}} \leq 2 \log n + \frac{2\rho (1+\epsilon)}{1- \epsilon} \sum_{i=1}^n \tau^{\operatorname{OL}}(\aa_i) =  O(d \eps^{-2} \log n \log \kappa^{\operatorname{OL}}).
    \]

    It thus remains to show that the event happens with probability at least $1-3n^{-1}$. To do so, notice that $N_i$ is $\mathcal{F}_i$-measurable. The scalar Freedman inequality from \Cref{lma:scalar_freedman} then yields that $N_n \leq 2 \log n + 2 \sum_{i=1}^n p_i$ with probability at least $1-n^{-2}$. We can finish the proof by a union bound: 
    \[
    \Prb[\mathcal{E}^c] \leq 
    \Prb[N_n > 2 \log n + 2\sum_{i=1}^n p_i] + \Prb[\exists i: (1- \eps) \AA_i^\top \AA_i \not \preceq \Tilde{\AA}_i^\top \Tilde{\AA}_i] 
    \leq n^{-2} + 2 n^{-1} \leq 3 n^{-1}.
    \]
    Finally, the working memory bound follows from the fact that one can compute $(\Tilde{\AA}_{i-1}^\top \Tilde{\AA}_{i-1} + \aa_i \aa_i^{\top})^\dag \aa_i$ using e.g. the conjugate gradient method without having to form the dense matrix $\Tilde{\AA}_{i-1}^\top \Tilde{\AA}_{i-1}$ explicitly. 
\end{proof}

\subsection{Proof of \Cref{thm:main}}\label{sec:thm_proof}

In order to prove the theorem, we have to employ the following custom matrix-concentration bound, which we will prove later in Section \ref{sec:concentration}. It deals with the issue that even the matrix $\XX_i^{\top} \XX_i$ is itself a random matrix due to the fact that the adversary might choose new rows dependent on the sampling outcomes of the previous rows.

\begin{restatable}{theorem}{concentrationTheorem}\label{lem:fixed-freedman}
Let $(\cF_i)_{i=0}^{n}$ be a filtration.  Let $\YY_0, \VV_0 = \matzero \in \mathbb{R}^{d \times d}$.  At index $i$, suppose $\CC_i \in \R^{d \times d}$ and $\WW_i \in \R^{d \times d}$ are $\cF_{i-1}$-measurable, $\CC_i$ satisfies that $\CC_i \CC_i^\top \preceq \II_d$,
and $\ZZ_i$ is an $\cF_i$-measurable symmetric random matrix such that
\[
        \E[\ZZ_i\mid\cF_{i-1}]=0,
        \qquad
        \lambda_{\max}(\ZZ_i)\le R,
        \qquad
        \E[\ZZ_i^2\mid\cF_{i-1}]\pe \WW_i.
\]
Define
\begin{equation}\label{eq:YV-recursion}
        \YY_i=\CC_i\YY_{i-1}\CC_i^\top+\ZZ_i,
        \qquad
        \VV_i=\CC_i\VV_{i-1}\CC_i^\top+\WW_i.
\end{equation}
Then, for every $u\ge0$,
\begin{equation}\label{eq:fixed-freedman}
        \Prb\!\left[\exists i \leq n:\lambda_{\max}(\YY_i)\ge u \land \VV_i \preceq \sigma^2 \II_d\right]
        \le d(n+1)\exp\!\left(-\frac{u^2/2}{\sigma^2+Ru/3}\right).
\end{equation}
\end{restatable}
Using this tool, we can now prove the main theorem.
\begin{proof}[Proof of \Cref{thm:main}]
Let $(\mathcal{F}_i)$ be the filtration where $\mathcal{F}_i$ includes all random bits used by the algorithm after the $i$th row has been processed. Note that we may assume that the adversary is deterministic. Consequently, once $\cF_{i-1}$ is fixed, $\aa_i$ is deterministic, and it is $\mathcal{F}_{i-1}$-measurable.

We now proceed to proving that $\Tilde{\AA}_i$ is a spectral approximation of $\AA_i$. To begin, let us set $\KK_i = \AA_i^\top \AA_i$, and $\Tilde{\KK}_i = \Tilde{\AA}_i^\top \Tilde{\AA}_i$. Further, we define the normalized error matrix
\begin{equation}\label{eq:E-def}
        \EE_i = \KK_i^{\dag/2}(\Tilde{\KK}_i - \KK_i)\KK_i^{\dag/2}.
\end{equation}
Our goal will be to show that $\norm{\EE_i}_{\mathrm{op}} \leq \eps$ with high probability. This implies that $(1- \eps) \KK_i \preceq \Tilde{\KK}_i \preceq (1+\eps) \KK_i$.

For the arriving row $\aa_i$, set
\[
        \vv_i = \KK_i^{-1/2}\aa_i\in\R^d,
        \qquad
        \norm{\vv_i}_2^2=\aa_i^\top \KK_i^{-1} \aa_i=\tau_i^{\operatorname{OL}},
\]
and
\[
        z_i=\frac{\xi_i}{p_i}-1,
        \qquad
        \ZZ_i=z_i \vv_i\vv_i^\top.
\]
Note that 
\begin{equation}\label{eq:long}
\EE_{i} = \KK_i^{\dag/2} \KK_{i-1}^{1/2} \EE_{i-1} \KK_{i-1}^{1/2} \KK_i^{\dag/2} + \ZZ_i.
\end{equation}
So let us define
\begin{equation}\label{eq:C-def}
        \CC_i:=\KK_i^{\dag/2}\KK_{i-1}^{1/2}.
\end{equation}
Notice that $\CC_i$ is $\mathcal{F}_{i-1}$-measurable, as the adversary chooses row $\aa_i$ based on our output $\Tilde{\AA}_{i-1}$. Since $\KK_{i-1}\pe \KK_i$,
\[
        \CC_i\CC_i^\top = \KK_i^{\dag/2}\KK_{i-1}\KK_i^{\dag/2}\pe \II_d.
\]
We can then write \Cref{eq:long} in the more compact form
\begin{equation}\label{eq:E-recursion}
        \EE_i=\CC_i\EE_{i-1}\CC_i^\top +\ZZ_i.
\end{equation}

The increment is conditionally mean zero:
\[
        \E[\ZZ_i\mid\cF_{i-1}]
        =\left(\frac{\E[\xi_i\mid\cF_{i-1}]}{p_i}-1\right)\vv_i\vv_i^\top
        =\matzero.
\]

We now need to introduce an additional stopping time, which deals with the fact that our current value $\tilde{\tau}_i$ might not be close to the value of $\tau^{\operatorname{OL}}(\aa_i)$. To do so, let $T$ be the first time $i$ where $ \Tilde{\AA}_i^\top \Tilde{\AA}_i \not \preceq (1+\epsilon) \AA_i^\top \AA_i$, and let $T = n$ on the complement of this event. We from now on consider the stopped martingale $\EE_i^{\operatorname{stop}} := \EE_{T \land i}$. Notice that $\norm{\EE_i}_{\mathrm{op}} \geq \epsilon \implies \norm{\EE_{i \land T}}_{\mathrm{op}} \geq \epsilon$, so it suffices to prove that $\norm{\EE_{i}^{\operatorname{stop}}}_{\mathrm{op}} \leq \epsilon$. 

To do so, we proceed by showing bounds on the parameters $R$ and $\sigma^2$ from \Cref{lem:fixed-freedman} for the stopped martingale $\EE_{i}^{\operatorname{stop}}$. To do so, start by noting that if $T \leq i-1$, then $\ZZ_i^{\operatorname{stop}} = \ZZ_{i-1}^{\operatorname{stop}}$, so in this case we have $\WW_i = 0$, and $R = 0$, and so beyond $T$, the values of $\sigma^2$ and $R$ will not change anymore. If $T \geq i-1$, then note first that this implies that $\Tilde{\AA}_{i-1}^\top \Tilde{\AA}_{i-1} \preceq (1+\epsilon) \AA_{i-1}^\top \AA_{i-1}$. Consequently, in this case, $\tilde{\tau}_{i \land T} \geq \rho (1+\epsilon) \frac{1}{1+\epsilon} \tau_i^{\operatorname{OL}} = \rho \tau_i^{\operatorname{OL}}$. Now if $p_i=1$, then $\xi_i=1$ almost surely and $\ZZ_i=0$.  If $p_i<1$, then
$p_i\geq\rho \cdot (1+\epsilon) \cdot \Tilde{\tau}_i \geq \rho \tau^{\operatorname{OL}}_i$.  In this case, when $\xi_i=1$ the scalar $z_i$ equals
$1/p_i-1$, and when $\xi_i=0$ it equals $-1$.  Therefore
\begin{align}
        \lambda_{\max}(\ZZ_i)
        &\le \left(\frac1{p_i}-1\right)\norm{\vv_i}_2^2
         \le \frac{\tau_i^{\operatorname{OL}}}{p_i}=\frac1\rho,\label{eq:upper-inc}\\
        \lambda_{\max}(-\ZZ_i)
        &\le \norm{\vv_i}_2^2=\tau_i^{\operatorname{OL}}\le\frac1\rho.\label{eq:lower-inc}
\end{align}
Thus both tails have increment bound $R=\rho^{-1}$.

For the conditional variance, use
\[
        (\vv_i\vv_i^\top)^2=\norm{\vv_i}_2^2\vv_i\vv_i^\top=\tau_i^{\operatorname{OL}}\vv_i\vv_i^\top.
\]
Also
\[
        \E[(z_i)^2\mid\cF_{i-1}]
        =p_i\left(\frac1{p_i}-1\right)^2+(1-p_i)
        =\frac{1-p_i}{p_i}
        \le \frac1{p_i}.
\]
Therefore
\begin{equation}\label{eq:variance-increment}
        \E[\ZZ_i^2\mid\cF_{i-1}]
        =\E[(z_i)^2\mid\cF_{i-1}] \cdot (\vv_i\vv_i^\top)^2 \pe \frac{\tau_i^{\operatorname{OL}}}{p_i} \vv_i \vv_i^\top 
        \pe \WW_i:=\frac1\rho \vv_i\vv_i^\top
\end{equation}
When $p_i=1$, then  $\E[(z_i)^2\mid\cF_{i-1}] = 0$, so the same bound is still valid.

Let
\[
        \VV_i=\CC_i\VV_{i-1}\CC_i^\top + \WW_i,
        \qquad \VV_0=\matzero.
\]
We next show that
\begin{equation}\label{eq:variance-budget}
        \VV_i\pe \frac1\rho \II\qquad\text{for every }i.
\end{equation}
Indeed, assuming it holds at $i-1$, then if $T \leq i-1$, it will also hold at $i$. But also if $T \geq i-1$, then
\begin{align*}
        \VV_i
        &\pe \frac1\rho \CC_i\CC_i^\top +\frac1\rho \vv_i\vv_i^\top \\
        &=\frac1\rho \KK_i^{\dag/2}\KK_{i-1}\KK_i^{\dag/2}
          +\frac1\rho \KK_i^{\dag/2}\aa_i\aa_i^\top \KK_i^{\dag/2} \\
        &=\frac1\rho \KK_i^{\dag/2}\KK_i\KK_i^{\dag/2}
         \preceq \frac1\rho \II.
\end{align*}
This proves \eqref{eq:variance-budget} by induction. 

Apply \Cref{lem:fixed-freedman} with the matrix process $\YY_i=\EE_i^{\operatorname{stop}}$ and with
\[
        R=\rho^{-1},
        \qquad \sigma^2=\rho^{-1},
        \qquad u=\eps.
\]
For any fixed $n$,
\begin{align}\label{eq:one-tail}
        \Prb\!\left[\exists i\le n:\lambda_{\max}(\EE_{i}^{\operatorname{stop}})\ge\eps\right]
        &\le d(n+1)\exp\!\left(
          -\frac{\eps^2/2}{\rho^{-1}+\eps(3\rho)^{-1}}
        \right) \\
        &\le d(n+1)\exp\!\left(-\frac38\eps^2\rho\right),
\end{align}
where $0<\eps\le1$ was used in the last line.  Applying the same bound to $-\EE_i$ gives
\begin{equation}\label{eq:two-tail}
        \Prb\!\left[\exists i\le n:\norm{\EE_{i}^{\operatorname{stop}}}_{\mathrm{op}}\ge\eps\right]
        \le 2d(n+1)\exp\!\left(-\frac38\eps^2\rho\right).
\end{equation}
Plugging in the value of $\rho = 8 \eps^{-2} \log n$ yields a bound of $2n^{-1}$. This finishes the proof of the spectral-accuracy part of Theorem \ref{thm:main}.

\end{proof}

\section{Custom Matrix Concentration Bound}\label{sec:concentration}


\subsection{Standard Tools}
We start by recalling the following standard tools
\begin{lemma}[Matrix MGF bound (see \cite{tropp2015introduction})]\label{lem:mgf}
Let $R>0$ and let $\XX$ be a symmetric random matrix with $\E[\XX]=0$ and
$\lambda_{\max}(\XX)\le R$ almost surely.  For every $0 < \theta < 3 / R$,
\begin{equation}\label{eq:mgf-bound}
        \log \E e^{\theta \XX}\pe
        g_R(\theta)\E[\XX^2],
        \qquad
        g_R(\theta)=
        \frac{\theta^2 / 2}{1 - \theta R /3}
\end{equation}
\end{lemma}

\begin{theorem}[Lieb (see \cite{tropp2015introduction})]\label{thm:lieb}
Let $\HH$ be a fixed symmetric matrix and let $\XX$ be a random symmetric matrix of the same
dimension.  Then
\begin{equation}\label{eq:lieb}
        \E\,\tr\exp(\HH+\XX)
        \le \tr\exp\!\left(\HH+\log\E e^{\XX} \right),
\end{equation}
whenever the expectations exist. 
\end{theorem}

\subsection{Trace Lemma}

The evolving normalization will replace a matrix $\HH$ by $\CC\HH\CC^\top$, where $\CC$ satisfies that $\CC\CC^\top \preceq \II_d$. The following lemma says that this operation cannot create many large positive
eigenvalues.  The extra additive $d$ accounts for eigenvalues that are nonpositive, since
$e^\lambda\le1$ for $\lambda\le0$.

\begin{lemma}\label{lma:contraction-trace}
Let $\HH \in \R^{d \times d}$ be a symmetric matrix, and let $\CC \in \R^{d \times d}$ be a matrix satisfying $\CC\CC^\top \pe \II$.  Then
\begin{equation}\label{eq:trace-contraction}
        \tr\exp(\CC\HH\CC^\top)\le d+\tr\exp(\HH).
\end{equation}
\end{lemma}

\begin{proof}
For a symmetric matrix $\YY$, let us denote by $N_{\YY}(a)$ the number of eigenvalues of $\YY$ that are strictly larger than $a$ (counting multiplicity).  We start by showing that
\begin{equation}\label{eq:counting}
        N_{\CC\HH\CC^\top}(a)\le N_{\HH}(a)
        \qquad\text{for every }a>0.
\end{equation}
By assumption that $\CC\CC^\top \preceq \II$, it holds for any $\xx \in \R^n$ that $\norm{\CC^\top \xx}_2^2 = \xx^\top \CC \CC^\top \xx \leq \xx^\top \xx = \norm{\xx}_2^2$. 
Let $\calS \subseteq \R^n$ be the subspace spanned by the eigenvectors of $\CC\HH\CC^\top$ corresponding to eigenvalues larger than $a$.
If $\xx \in \calS \setminus\{0\}$, then
\[
        (\CC^\top \xx)^\top \HH \CC^\top \xx
        =\xx^\top \CC\HH\CC^\top \xx
        >a\norm{\xx}_2^2 \ge \norm{\CC^\top \xx}_2^2.
\]
Note that this implies on one hand that $\dim{\CC^\top \calS} = \dim{\calS}$, and, consequently, by the variational characterization of eigenvalues, also that $\HH$ must have at least $\dim{\calS}$ eigenvalues strictly larger than $a$, which proves \eqref{eq:counting}.

For a scalar $\lambda \geq 0$, note the elementary equality
\[
        e^\lambda-1=\int_0^\infty e^a\mathbf 1_{\{a<\lambda\}}\,da.
\]
Summing over the positive eigenvalues of $\YY$ gives the identity
\[
        \sum_{i: \lambda_i \geq 0} (e^{\lambda_i(\YY)}-1)
        = \int_0^\infty \sum_{i: \lambda_i \geq 0} e^a\mathbf 1_{\{a<\lambda_i\}}\,da = \int_0^\infty e^a N_{\YY}(a)\,da.
\]
Using \eqref{eq:counting} and that $e^{\lambda} \leq 1$ for $\lambda \leq 0$, we have that
\[
        \tr e^{\CC\HH\CC^\top}
        \le d+\sum_{i: \lambda_i \geq 0}(e^{\lambda_i(\CC\HH\CC^\top)}-1) 
        \le d+\sum_{i: \lambda_i \geq 0}(e^{\lambda_i(\HH)}-1)
        \le d+\tr(e^{\HH}).
\]
This proves the lemma.
\end{proof}

\subsection{Contractive Freedman Inequality}

We now prove the matrix concentration lemma used in the sampling argument.  The proof is a slightly modified version of the classic matrix Freedman concentration inequality. 

\concentrationTheorem*

\begin{proof}
Fix $0 < \theta < R / 3$ and set
\[
        g_R(\theta)=
        \frac{\theta^2 / 2}{1 - \theta R /3}.
\]
 Define the potential
\[
        \Psi_i=\tr\exp(\theta \YY_i-g_R(\theta)\VV_i).
\]
We will start by proving
\begin{equation}\label{eq:psi-step}
        \E[\Psi_i\mid\cF_{i-1}]\le \Psi_{i-1}+d.
\end{equation}
Condition on $\cF_{i-1}$.  At this point $\CC_i$, $\YY_{i-1}$, $\VV_{i-1}$, and $\WW_i$ are fixed, while
$\ZZ_i$ is the only fresh randomness.  Write
\[
        \HH_i=\theta \CC_i\YY_{i-1}\CC_i^\top
             -g_R(\theta)\CC_i\VV_{i-1}\CC_i^\top -g_R(\theta)\WW_i.
\]
By \Cref{lem:mgf} and the variance assumption,
\[
        \log\E[e^{\theta \ZZ_i}\mid\cF_{i-1}]
        \pe g_R(\theta)\E[\ZZ_i^2\mid\cF_{i-1}]
        \pe g_R(\theta)\WW_i.
\]
Using Lieb's inequality from \Cref{thm:lieb}, then the monotonicity of $\tr e^{\HH}$, and finally \Cref{lma:contraction-trace},
\begin{align*}
        \E[\Psi_i\mid\cF_{i-1}]
        &=\E\left[\tr\exp(\HH_i+\theta \ZZ_i)\mid\cF_{i-1}\right] \\
        &\le \tr\exp\!\left(\HH_i+\log\E[e^{\theta \ZZ_i}\mid\cF_{i-1}]\right) \\
        &\le \tr\exp\!\left(\theta \CC_i\YY_{i-1}\CC_i^\top
             -g_R(\theta)\CC_i\VV_{i-1}\CC_i^\top\right) \\
        &=\tr\exp\!\left(\CC_i(\theta \YY_{i-1}-g_R(\theta)\VV_{i-1})\CC_i^\top \right) \\
        &\le d+\tr\exp(\theta \YY_{i-1}-g_R(\theta)\VV_{i-1})
         =d+\Psi_{i-1}.
\end{align*}
This proves \eqref{eq:psi-step}.

Note that this means $\Psi_i$ is not quite a (super-)martingale yet. However, by a slight modification we can turn it into a super-martingale. To do so, set
\[
        Q_i=\Psi_i+d(T-i),\qquad 0\le i\le n.
\]
Then $(Q_i)$ is a nonnegative super-martingale.
Hence, for every stopping time $T \le n$, we have that
\begin{equation}\label{eq:optional-stop}
        \E Q_T \le Q_0=d(n+1),
\end{equation}
by the optional stopping theorem and because $\YY_0=\VV_0=\matzero$ and therefore $\Psi_0=\tr(\II_d)=d$. 

Let $T$ be the first index $i\le n$ at which $\lambda_{\max}(\YY_i)\ge u$ and $\VV_{i} \preceq \sigma^2 \II_d$, and set $T=n$ on the
complement of this event.  On the bad event, it follows that
\[
        \lambda_{\max}(\theta \YY_{T}-g_R(\theta) \VV_{T})
        \ge \theta u-g_R(\theta)\sigma^2.
\]
Consequently, on the bad event,
\[
        \Psi_{T} \ge \exp(\theta u-g_R(\theta)\sigma^2).
\]
Combining this with \eqref{eq:optional-stop} gives
\[
        \Prb\!\left[\exists i\le n:\lambda_{\max}(\YY_i)\ge u \land \VV_i \preceq \sigma^2 I_d\right]
        \le d(n+1)\exp(-\theta u+g_R(\theta)\sigma^2).
\]
It remains to note that choosing $\theta = u / (\sigma^2 + R u / 3) < 3/R$ and plugging it into the above bound yields that
$$\Prb\!\left[\exists i \leq n:\lambda_{\max}(\YY_i)\ge u \land \VV_i \preceq \sigma^2 I_d\right]
        \le d(n+1)\exp\!\left(-\frac{u^2/2}{\sigma^2+Ru/3}\right).$$
This proves the theorem.
\end{proof}

\section{Online Sparsification for Graphs}\label{sec:graphs}
In this section, we present a fast and memory efficient implementation of the online sparsification algorithm from \Cref{algo:main_algo}, for the special case where the rows $\aa_i$ arriving to use are edges of a weighted undirected graph $G = (V, E, w)$. To be more precise, our online algorithm maintains a nearly-linear sized sparsifier using only nearly-linear memory in the number of vertices $|V|$. The total run-time of our algorithm is also nearly-linear in the length of the stream $m$. Our main result is summarized in \Cref{thm:graph-sparsifier}.

Throughout this section, we assume that the graph $G$ has weights in $[1, W]$. We denote by $G_i$ the graph $G$ after the first $i$ edges have been revealed to us. We emphasize again that the graph $G_i$ is itself random and depends on the sampling outcomes of our online algorithm.

We denote by $\Tilde{G}_i$ the maintained online sparsifier of $G_i$ after processing the $i$'th edge insertion. Given a graph $G$ with incidence matrix $B \in \R^{E \times V}$ and diagonal weight matrix $\WW = \diag(\ww)$, we denote by $\LL = \BB^{\top} \WW \BB$ the Laplacian matrix associated with $G$. In spectral graph sparsification, we are interested in finding a sparse weight vector $\widetilde{\ww}$ such that $\Tilde{\LL} = \BB^T \widetilde{\WW} \BB$ is a $(1 \pm \epsilon)$ spectral approximation of $\LL$.

\begin{theorem} \label{thm:graph-sparsifier}
Let $G = (V, E, w)$ be a graph with edge weights in $[1, W]$ whose edges $\bb_{e_1}, \bb_{e_2} \ldots, \bb_{e_m}$ arrive to us in a stream, and let $\epsilon \in (0, 1)$, $c > 0$ be parameters. Then, there is an adversarially robust online algorithm that maintains an online sparsifier $\Tilde{G}_i \subseteq G_i$ such that for all $i \in [m]$,
\[
    (1-\eps)\LL_i \pe \Tilde{\LL}_i \pe (1+\eps) \LL_i
\]
where $\LL_i$ and $\Tilde{\LL}_i$ are Laplacians of $G_i$ and $\Tilde{G}_i$, the graph and the maintained sparsifier after $i$ updates.
Throughout all updates, the algorithm satisfies the size of $\Tilde{G}$ is $\Tilde{O}(n\eps^{-2} \log W)$, requires $\Tilde{O}(n \eps^{-2} \log W)$ memory, and takes total time $\Tilde{O}((m+n) \eps^{-2} \log W)$.
All results hold with probability $1-\frac{1}{m^c}$ over the entire stream.
\end{theorem}

Our algorithm is based on the spectral sparsification algorithm of \cite{KoutisX16} and its fully dynamic variant \cite{AbrahamDKKP16}.
In their algorithms, a bundle of polylogarithmically many spanners (see \Cref{def:tbundle}) is used to certify a large number of edges with low leverage scores.
Each such half-sparsification step then reduces the number of edges by roughly half (hence the name), resulting in a total of $O(\log m)$ such steps.
To ensure that this bundle of spanners changes in a controlled way as new edges are inserted, we need to consider \emph{online} spanners for this task.

\begin{definition}[$t$-bundle spanners, \cite{KoutisX16,AbrahamDKKP16}] \label{def:tbundle}
   A $t$-bundle $\alpha$-spanner of an undirected graph $G$ is a union of $B = \bigcup_{j=1}^t H_j$ of a sequence of graphs $H_1, \ldots, H_t$ such that, for every $1 \le j \le t$, $H_j$ is an $\alpha$-spanner of $G \setminus \bigcup_{k=1}^{j-1} H_k$.
\end{definition}

The following lemma shows why $t$-bundle spanners are so useful: it allows us to deduce upper bounds on the leverage scores of edges that are in $G \setminus B$.

\begin{lemma}[\cite{KoutisX16}] \label{lem:bundle-lev}
Let $G$ be a graph and $B$ be a $t$-bundle $\alpha$-spanner of $G$.
For every edge $e \in G \setminus B$,
\[
    \tau_G(e) := \ww_e \bb_e^\top \LL_G^{\dag} \bb_e \le \frac{\alpha}{t}
\]
\end{lemma}

We can now present our algorithm which assumes black-box access to a deterministic online spanner algorithm. For the rest of this section, we let $\calA$ be a deterministic online $\alpha$-spanner algorithm such that over $m$ insertions, the function $N_{\calA}(|V|,\alpha,W) \ge |V|$ is an upper bound on the number of edges of the spanner, which we assume here does not depend on the edge insertions $m$. This is the case for most online spanner algorithms. Additionally, we set $U_{\calA}(|V|,\alpha, W, m)$ to be the total work, and $M_{\calA}(|V|,\alpha,W,m)$ to be its total memory.
Note that it is always satisfied that $M_{\calA}(|V|,\alpha,W,m) \ge N_{\calA}(|V|,\alpha,W)$.

\begin{lemma} \label{lem:tbundle-online}
For every $t \ge 1$, given a deterministic online $\alpha$-spanner algorithm $\calA$, there is a deterministic online algorithm for maintaining a $t$-bundle $\alpha$-spanner $B$ of a weighted graph $G$ on $n$ vertices and bounded edge weight $W$ of size $O(t \cdot N_{\calA}(|V|,\alpha,W))$.

The algorithm has total running time $O(t \cdot U_{\calA}(|V|,\alpha,W,m))$ and uses memory $O(t \cdot M_{\calA}(|V|,\alpha,W,m))$, both over $m$ updates.
\end{lemma}
\begin{proof}
We run $t$ copies of the algorithm $\calA$ and let $H_i$ denote the spanner maintained by the $i$th copy.
For each $i$, set $H_i$ be an $\alpha$-spanner of $G \setminus \bigcup_{j\le i-1} H_i$.
Then, the $t$-bundle spanner is simply a union of these spanners, i.e., $B = \bigcup_{i=1}^t H_i$.

Given an edge insertion of $e$, the algorithm simply inserts the edge in a sequence from $i=1$ to $t$ until the edge is added to the spanner $H_i$. Since $\calA$ is online, no edge is ever removed from $H_i$ for every $i$. Thus, $B$ is also online and, for each $i$, the spanner $H_i$ receives at most $m$ edges insertions.

Our claimed size, total time, and memory guarantees then directly follows. It is also easy to see that the algorithm is deterministic as long as $\cal A$ is.
\end{proof}

\begin{algorithm}[!ht]
\caption{$\textsc{OnlineGraphSparsification}(G, \eps)$}
\label{algo:graphsparsify}
\DontPrintSemicolon
\textbf{Input}: Adversarial stream of weighted edges $e_1, e_2, \ldots, e_m \in V \times V$ with edge weights $w_{e_1}, \ldots w_{e_m} \in [1,W]$.\;
\textbf{Output}: $(1 \pm \eps)$-spectral approximation $\Tilde{G}_i$ of $G_i$\; 
$L \gets \max\left(2,\ceil{\log_4 \frac{m}{|V|}}\right)$, $t \gets \ceil{80 (c+3)\alpha L^2 \eps^{-2}\log(|V|Lm)}$\;
Initialize online $t$-bundle spanner $B^{(j)}$ for every $j \in [L]$ \hfill \Comment{\Cref{lem:tbundle-online}}
Let $\tilde{G}_i$ be $G_i$ reweighted by $\Tilde{\ww}_i$ always.\;
\For{$i = 1, 2, \ldots, m$}{
    $\tilde{w}_{e_i} \gets w_{e_i}$\;
    \For{$j = 1,2 \ldots, L$}{ \label{line:graphsparsifty_inner_start}
        $B^{(j)} \gets B^{(j)}$ with edge update $e_i$ and current weight $\tilde{w}_{e_i}$\;
        \If{$e_i \not\in B_j$}{
            $\tilde{w}_{e_i} \gets 4\tilde{w}_{e_i}$ with probability $1/4$\;
            Otherwise, $\tilde{w}_{e_i} \gets 0$\;
        }
        \If{$e_i \in B_j$ or $\tilde{w}_i = 0$}{
            \Break
        }
    }
}
\end{algorithm}

Algorithm~\ref{algo:graphsparsify} is the pseudo-code of our online spectral graph sparsification algorithm.
In our main theorem regarding graph sparsification, \Cref{thm:graph-sparsifier}, we use for $\calA$ the following deterministic online spanner algorithm in \Cref{thm:det_spanner}. A proof of this result using results on incremental APSP data structures is given in Appendix \ref{sec:spanner}.
\begin{theorem}\label{thm:det_spanner}
    There is a deterministic online algorithm for computing a $\Otil(1)$-spanner $H$ of $G$. The spanner $H$ has size $O(n \log W)$, and and the algorithm takes total time $O(m \log W \log \log n + n \log W \log^6 n \log \log n)$. It requires $O(n \log W \log^6 n \log \log n)$ memory.
\end{theorem}

At a high level, Algorithm~\ref{algo:graphsparsify} remains faithful to the dynamic spectral sparsification algorithm by \cite{AbrahamDKKP16}. Conceptually, for each $0 \leq i \leq m$ and $j \in [L]$, we let $\ww_{e_i}^{(j)}$ be the edge weight $\tilde{w}_{e_i}$ after the $j$'th iteration of the inner loop at line~\ref{line:graphsparsifty_inner_start}. This allows us to define, for every $j \in [L]$, the graph $G^{(j)}$ to be the graph $G$ but with weights $\ww^{(j)}$. In this notation, we can write for our final sparsifier $\Tilde{G}$ that $\Tilde{G} = G^{(L)}$, and the graph $G$ is simply equal to $G^{(0)}$. We can then view the bundles $B^{(j)}$ to be online spanner bundles of the graphs $G^{(j-1)}$. As usual, we can then write for example $G^{(j)}_i$ to denote the graph $G^{(j)}$ after the first $i$ edges of the stream have been processed.

Our algorithm can be viewed as repeated half-sparsification of the $G^{(j)}$s.
We show in \Cref{lem:graph-levels-approx} that each $G^{(j+1)} \cup B^{(j)}$ is a good spectral approximation of $G^{(j)}$ throughout all updates. 

The following \Cref{lem:graph-online-levels} summarizes a few basic properties of our algorithm that are clear by construction.

\begin{lemma} \label{lem:graph-online-levels}
For all $j \in [L]$, the graph $G^{(j)}$ only receives edges insertions.
Once an edge is added to $G^{(j)}$, this edge remains in $G^{(j)}$ with the same edge weight.
\end{lemma}

\begin{lemma} \label{lem:graph-levels-approx}
Given at most $m \ge 10$ adversarial edge insertions, for any $j \in [L]$, the graph $G^{(j)} \cup B^{(j)}$ is a $(1\pm \frac{\eps}{2L})$-spectral sparsifier of $G^{(j-1)}$ with probability $1-\frac{1}{2Lm^{c+1}}$ throughout all $m$ updates. 
\end{lemma}
\begin{proof}
To prove this statement, we will use the framework of \Cref{algo:main_algo} and \Cref{thm:main}. Let $e_i$ be the newly inserted edge. We start by defining 
$$\tau_i^{(j-1)} := \ww_{e_i}^{(j-1)} \bb_{e_i}^{\top} \LL_{G^{(j-1)}_i}^{\dag} \bb_{e_i}$$
to be the leverage score of the edge $e_i$ in the graph $G^{(j-1)}_i$. Now in order to prove the lemma, the only thing we need to show is that once a new edge $e_i$ gets added (adversarially) to the stream, it gets sampled into the graph $G^{(j)} \cup B^{(j)}$ with probability at least
$$
\min \setof{1, \rho \cdot \tau_i^{(j-1)}}
$$
for $\rho = O(L^2 \epsilon^{-2} \log n)$ large enough. If $e_i$ gets inserted into $B^{(j)}$, then this corresponds to sampling with probability $1$, in which case we have certainly sampled appropriately. On the other hand, if $e_i$ does not get inserted into $B^{(j)}$, then by \Cref{lem:bundle-lev} it must be the case that 
$$
\tau_i^{(j-1)} \leq \frac{\alpha}{t} \leq \frac{\epsilon^2}{80(c+3) L^2 \log (|V| L m)}.
$$
Consequently, we have that 
$$
\rho \tau_i^{(j-1)} \leq \frac{1}{4},
$$
and since we sample the edge $e_i$ into $G^{(j)}$ with probability exactly $1/4$, this constitutes valid leverage-score oversampling. We can now apply \Cref{thm:main} to conclude that $G^{(j)} \cup B^{(j)}$ is a $(1 \pm \frac{\varepsilon}{2L})$-spectral sparsifier of $G^{(j-1)}$ throughout all $m$ updates with probability at least $1- \frac{1}{2Lm^{c+1}}$. 

\end{proof}

The following \Cref{lem:graph-approx} follows directly from \Cref{lem:graph-levels-approx}.
\begin{lemma} \label{lem:graph-approx}
Given at most $m \ge 10$ adversarial edge insertions, Algorithm~\ref{algo:graphsparsify} maintains a weighted subgraph $\Tilde{G}_i$ such that with probability $1-\frac{1}{2m^{c+1}}$ for every $i \in [m]$, 
\[
(1-\eps) \LL_i \pe  \Tilde{\LL}_i \pe (1+\eps) \LL_i.
\]
\end{lemma}
\begin{proof}
    Note that $\Tilde{G}_i = G^{(L+1)}_i \cup \bigcup_{j=1}^L B^{(j)}_i$. Now assume that for all $j$, the graphs $G^{(j)} \cup B^{(j)}$ are $(1 \pm \frac{\varepsilon}{2L})$-spectral sparsifiers of $G^{(j-1)}$. Then it follows that $\Tilde{G}$ is a $(1 \pm \frac{\varepsilon}{2L})^{L}$-spectral sparsifier of $G^{(0)} = G$. We can conclude by noting that $(1+\frac{\epsilon}{2L})^L \leq (1+\epsilon)$, and $(1- \frac{\epsilon}{2L})^L \geq (1- \epsilon)$.
\end{proof}

\begin{lemma} \label{lem:graph-size}
Given at most $m \ge 10$ adversarial edge insertions, Algorithm~\ref{algo:graphsparsify} maintains a weighted subgraph $\Tilde{G}_i$ such that with probability $1-\frac{1}{2m^{c+1}}$ for every $i \in [m]$, the number of edges 
\begin{equation} \label{eq:graph-size}
    |E(\Tilde{G})| = O\left(Lt \cdot N_{\calA}(|V|,\alpha,W,m) + \frac{m}{4^L} + \log m\right).
\end{equation}
\end{lemma}
\begin{proof}
The first term in \eqref{eq:graph-size} is clear from the $L$ levels of $t$-bundle spanners.
We can then focus on the trailing terms.

For an edge insertion $e_i$, instead of flipping a biased coin for each iteration $j=1$ to $L$ separately, we can sample immediately after its insertion a uniform random variable $X_i \in [0,1]$.
If, at iteration $j$, $e_i$ is not added to $B_j$, then, we can check whether the biased sample is successful by comparing $X_i$ and $\frac{1}{4^j}$.
We have $e_i \in \Tilde{G}_k \setminus \bigcup_{j=1}^L B^{(j)}_k$ for all $k \ge i$ only if $X_i < \frac{1}{4^L}$, which is successful with probability $\frac{1}{4^L}$.
This further simplifies to a sampling a coin $\xi_i \sim \operatorname{Ber}(\frac{1}{4^L})$ conditionally on all previous insertions and random samples and counting $N_i = \sum_{k=1}^i \xi_i$.
Similar to the proof of \Cref{thm:main}, the scalar Freedman inequality from \Cref{lma:scalar_freedman} gives with probability $1-\frac{1}{2m^{c+1}}$ that
\[
    N_i \le N_m \le  4(c+2) \left(\log m + \frac{m}{4^L} \right),
\]
which completes our proof.
\end{proof}

\begin{theorem} \label{thm:graph-main}
There is an algorithm (Algorithm~\ref{algo:graphsparsify}) that, for any fixed constant $c \ge 1$, given an $\eps \in (0,1)$, an adversarial stream of $m$ edges $E$ of a graph $G$ supported on $V$ and edge weights $\ww \in [1,W]$ for fixed $W \ge 1$, such that for all $i \in [m]$, it maintains a reweighted subgraph $\Tilde{G}$ of $G$ with edge weights $\Tilde{\ww}$.
The algorithm satisfies with probability at least $1-\frac{1}{m^c}$ over the entire stream that 
\begin{enumerate}
\item \textbf{$(1\pm\eps)$-spectral approximation:} \label{thm:graph-main:item:approx}
\[
    (1-\eps) \LL_i \pe \Tilde{\LL}_i \pe (1+\eps) \LL_i.
\]
\item \textbf{Size:} the number of edges in $\Tilde{G}$ is at most \label{thm:graph-main:item:size}
\[
    O\left(\alpha \eps^{-2} \log(m|V|) \log^3\frac{m}{|V|} \cdot N_{\calA}(|V|,\alpha,W)\right).
\]
\item \textbf{Total time:} the amortized update time is \label{thm:graph-main:item:time}
\[
    O\left(\alpha \eps^{-2} \log(m|V|) \log^3\frac{m}{|V|} \cdot U_{\calA}(|V|,\alpha,W,m)\right). 
\]
\item \textbf{Memory:} the required memory of the algorithm is \label{thm:graph-main:item:mem}
\[
    O\left(\alpha \eps^{-2} \log(m|V|) \log^3\frac{m}{|V|} \cdot M_{\calA}(|V|,\alpha,W,m)\right). 
\]
\end{enumerate}
\end{theorem}
\begin{proof}[Proof of \Cref{thm:graph-main}]
The spectral approximation guarantee (\Cref{thm:graph-main:item:approx}) holds with probability $1-\frac{1}{2m^{c+1}}$ by \Cref{lem:graph-approx}.
Note that when $m<10$, no sparsification is required.
By our choices of 
\begin{align*}
L &= \max\left(1,\ceil{\log_4 \frac{m}{|V|}}\right) \ge \max\left(1, \log_4 \frac{m}{|V|} \right)\\
t &= \ceil{80 (c+3)\alpha L^2 \eps^{-2}\log(|V|Lm)} = O\left(\alpha \eps^{-2} \log^2(\frac{m}{|V|}) \log(m|V|)\right),
\end{align*}
in Algorithm~\ref{algo:graphsparsify}, the size (\Cref{thm:graph-main:item:size}) and update time (\Cref{thm:graph-main:item:time}) bounds follow by \Cref{lem:graph-size,lem:tbundle-online} with probability $1-\frac{1}{2m^{c+1}}$.
In terms of memory, our algorithm stores (a) $L$ levels of $t$-bundle spanners and (b) the edges retained in the graph that are not within any $t$-bundle spanners.
For (a), we can bound the required memory using \Cref{lem:tbundle-online} again and by noting that each $t$-bundle can only ever receive $m$ edge insertions (\Cref{lem:graph-online-levels}).
The memory for (b) can be bounded simply using \Cref{thm:graph-main:item:size}.
Since $M_{\calA} \ge N_{\calA}$ under the same parameters, the memory of (b) is subsumed in asymptotics, giving us the overall memory bound in \Cref{thm:graph-main:item:mem}.
Finally, all guarantees are satisfied with probability $1-\frac{1}{m^{c+1}}$ by union bound.
\end{proof}

\Cref{thm:graph-sparsifier} follows from \Cref{thm:graph-main} by substituting the deterministic online spanner algorithm in \Cref{thm:det_spanner} for $\calA$ and setting $\alpha = \Tilde{O}(1)$.

\newpage
\printbibliography

\newpage
\appendix
\section{Deterministic Online Spanner} \label{sec:spanner}
In this section, we will briefly describe a very simple algorithm satisfying the guarantees of \Cref{thm:det_spanner}. In order to so, we will borrow tools from the incremental/dynamic graph algorithms literature. In this literature, one is concerned, as in the streaming/online model, with a graph $G$ whose edges $e_i$ are revealed to the algorithm one at a time. One is then interested in maintaining query access to some quantities of interest in the current version of the graph, $G_i$. However, the difference to streaming algorithms is that in dynamic graph algorithms the main quantity of interest is the total computational work performed, not necessarily the working memory that the algorithm requires. 

Our online algorithm works by combining the classical greedy spanner (\cite{althofer1993spars}) with the incremental APSP data structure from \cite{incremental_APSP}. The algorithm is inspired from both the slow dynamic spanner algorithms presented in \cite{low_recourse_spanner}, and the fast versions presented in \cite{kyng2025simple}.


We first introduce incremental APSP data structures, where we say that a graph $G$ is incremental if it changes by edge insertions over time.

\begin{definition}[Incremental APSP] \label{def:apsp}
    For an initially empty incremental graph $G = (V, E)$ a $\gamma$-approximate incremental APSP data structure supports the following operations.
    \begin{itemize}
        \item $\textsc{AddEdge}(u,v)$: Add edge $(u, v)$ to $G$.
        \item $\textsc{ApxDist}(u, v)$: Returns a distance estimate $\widetilde{\dist}(u,v)$ such that
        \begin{equation*}
            \dist_G(u, v) \leq \widetilde{\dist}(u,v) \leq \gamma \cdot \dist_G(u,v). 
        \end{equation*}
    \end{itemize}
\end{definition}

\begin{theorem}[\cite{incremental_APSP}]\label{thm:det_apsp}
    For an incremental graph $G$, there is a deterministic $\Otil(1)$-approximate incremental APSP data structure. After processing $m$ calls to $\textsc{AddEdge}$ and $q$ calls to $\textsc{ApxDist}$, the total work it will have done is $O(m\log n \log \log n  + q \log \log n + n \log^6 n \log \log n)$. Moreover, the memory required by the algorithm is $O(m \log n \log \log n + n \log^6 n \log \log n)$.
\end{theorem}

Below is the pseudo-code of our algorithm. It is just the classic spanner from \cite{althofer1993spars}, but sped up by using an Incremental APSP data structure inside of it.

\begin{algorithm}[!ht]
\caption{$\textsc{OnlineSpanner}()$}
\label{algo:simple_spanner}
\SetKwProg{Globals}{global variables}{}{}
\SetKwProg{Proc}{procedure}{}{}
\Proc{$\textsc{Initialize}()$}{
    $H \gets (V, \emptyset)$\;
    Let $\mathcal{D}_{H}$ be the $\Otil(1)$-approximate APSP datastructure from \Cref{thm:det_apsp}.
}
\Proc{$\textsc{InsertEdge}(e = (u, v))$}{
    \label{algo:simple_spanner:if}\If{$\mathcal{D}_{H}.\textsc{ApxDist}(u,v) > \gamma \cdot 2 \cdot \log n$}{
        $H \gets H \cup (e)$; $\mathcal{D}_{H}.\textsc{AddEdge}(e)$;
    }
}
\end{algorithm}
Before we start the analysis, we need the following folklore observation. See e.g. \cite{althofer1993spars} for a proof.

\begin{observation}[Folklore] \label{obs:girth}
    Every graph $H = (V, E_H)$ with girth $2\delta + 1$ contains at most $2|V|^{1 + \frac{1}{\delta}}$ edges.
\end{observation}

\begin{proof}[Proof of \Cref{thm:det_spanner}]
First note that by using weight buckets $[2^i, 2^{i+1})$, we can reduce an online algorithm for unweighted graphs $G$ to an online algorithm for weighted graphs with weights in $[1, W]$ with just an additional overhead of $\log W$ in its run-time, memory-usage, and number of edges. 

Thus, it remains to analyze the online spanner for unweighted graphs that is described in \Cref{algo:simple_spanner}.

We start with correctness. That the algorithm is online is clear from its description. Moreover, because \Cref{thm:det_apsp} is deterministic, so is \Cref{algo:simple_spanner}. In order to bound the stretch $H$, it suffices to show that for every edge $(u, v) \in E(G)$, we have that 
$$\dist_H(u, v) \leq 2 \gamma \log n \cdot \dist_G(u, v).$$ But this condition must be satisfied at all times, because if at the time of insertion it hadn't been satisfied, then the edge $e$ would have been inserted into $H$ in Line $6$ of \Cref{algo:simple_spanner:if}. Note that as at any time we have $\dist_G(u, v) \geq \mathcal{D}_H.\textsc{ApxDist}(u, v) / \gamma$, it must be that $H$ has girth at least $2 \log n + 1$ at all times. Consequently, from \Cref{obs:girth}, it follows that $|E(H)| \leq 2n$. 

The runtime and memory bound follows because the algorithm calls $\mathcal{D}_H.\textsc{ApxDist}()$ exactly $m$ times, and calls $\mathcal{D}_H.\textsc{AddEdge}()$ at most $2n$ many times. The guarantees then follow from \Cref{thm:det_apsp}.
\end{proof}

\end{document}